\documentclass[runningheads]{llncs}
\usepackage[T1]{fontenc}
\usepackage{graphicx}

\usepackage{ellipsis, ragged2e}
\usepackage[l2tabu, orthodox]{nag}
\usepackage[english]{babel}
\usepackage{lmodern}
\usepackage[final]{microtype}

\usepackage[utf8]{inputenc}
\usepackage{amsfonts,amsmath}
\usepackage{enumerate}
\usepackage{csquotes}
\usepackage{xcolor}
\usepackage{paralist}
\usepackage{hyperref}
\usepackage{enumitem}
\usepackage[capitalize]{cleveref}

\input{0_macros}

\usepackage{booktabs}
\usepackage{changepage}

\usepackage{framed}
\usepackage{tikz}
\usepackage{tikz,pgffor}
\usepackage{pgfplots}
\pgfplotsset{compat=1.16}
\usetikzlibrary{shapes,positioning,arrows,arrows.meta,calc,automata,matrix,fit,backgrounds}

\tikzstyle{ang}=[regular polygon, regular polygon sides = 3,draw,inner sep=0pt,minimum size=6mm, yshift = -0.75 mm]
\tikzstyle{dem}=[shape=diamond,draw,inner sep=0pt,minimum size=6mm]
\tikzstyle{ran}=[shape=circle,draw,inner sep=0pt,minimum size=6mm]
\tikzstyle{det}=[shape=rectangle,draw,inner sep=0pt,minimum size=5mm]
\tikzstyle{tran}=[draw,->,>=stealth, rounded corners]
\tikzstyle{pt}=[shape=circle,draw,inner sep=0pt,minimum size=0mm]
\tikzstyle{trannoarrow}=[draw,>=stealth, rounded corners]

\tikzstyle{state}+=[minimum size = 6mm, inner sep=0,outer sep=1]
\colorlet{disabled}{lightgray}
\tikzstyle{state}=[draw,rectangle,inner sep=5pt,rounded corners=2pt]
\tikzstyle{action}=[font=\small,inner sep=0pt,outer sep=3pt]
\tikzstyle{actionnode}=[circle,draw=black,fill=black,minimum size=1mm,inner sep=0,outer sep=0]
\tikzstyle{actionedge}=[draw,-]
\tikzstyle{prob}=[font=\scriptsize,inner sep=0pt,outer sep=1pt]
\tikzstyle{probedge}=[draw,->]
\tikzstyle{directedge}=[draw,->]
\tikzset{chainarrow/.tip={Stealth[length=3pt]}}
\tikzset{>=chainarrow}

\newtoggle{arxiv}
\toggletrue{arxiv}
\iftoggle{arxiv}{}{

\let\llncssubparagraph\subparagraph
\let\subparagraph\paragraph
\usepackage{titlesec}
\let\subparagraph\llncssubparagraph

\titlespacing*{\section}{0pt}{2ex plus 1ex minus 0.5ex}{1ex plus 0.5ex minus 0.5ex}
\titlespacing*{\subsection}{0pt}{1.5ex plus 0.5ex minus 0.5ex}{1ex plus 0.5ex minus 0.5ex}
\titlespacing*{\subsubsection}{0pt}{1ex plus 0.5ex minus 0.25ex}{1ex}
\titlespacing*{\paragraph}{0pt}{0.75ex plus 0.5ex minus 0ex}{1ex}

\setlength{\textfloatsep}{0.2\textfloatsep} 
\setlength{\intextsep}{0.2\intextsep} 
\setlength{\floatsep}{0.5\floatsep} 
\setlength{\abovecaptionskip}{0.15\abovecaptionskip}
\setlength{\belowcaptionskip}{0.15\belowcaptionskip}

\setcounter{topnumber}{2}
\setcounter{totalnumber}{3}

\AtBeginDocument{
\setlength{\abovedisplayskip}{0.4\abovedisplayskip}
\setlength{\abovedisplayshortskip}{0.4\abovedisplayshortskip}
\setlength{\belowdisplayskip}{0.4\belowdisplayskip}
\setlength{\belowdisplayshortskip}{0.4\belowdisplayshortskip}
}

}

\newcommand{\qee}{\hfill\ensuremath{\bigtriangleup}}
\newcommand{\algomem}{\ensuremath{\mathsf{AlgMemLess}}}
\newcommand{\algogen}{\ensuremath{\mathsf{AlgDist}}}

\usepackage{todonotes}

\begin{document}
%
\title{MDPs as Distribution Transformers: Affine Invariant Synthesis for Safety Objectives\thanks{This work was supported in part by the ERC CoG 863818 (FoRM-SMArt) and the European Union’s Horizon 2020 research and innovation programme under the Marie Skłodowska-Curie Grant Agreement No.~665385 as well as DST/CEFIPRA/INRIA project EQuaVE and SERB Matrices grant MTR/2018/00074.}}
\titlerunning{Invariant Synthesis for Affine Safety Objectives in MDPs}
\author{S.~Akshay\inst{1} \and
	Krishnendu Chatterjee\inst{2} \and
	Tobias Meggendorfer\inst{3,}\inst{2} \and
	\DJ{}or\dj{}e \v{Z}ikeli\'c \inst{2}}
\authorrunning{S.~Akshay, K.~Chatterjee, T.~Meggendorfer, \DJ{}.~\v{Z}ikeli\'c}
%
\institute{
	Indian Institute of Technology Bombay, India \and
	Institute of Science and Technology Austria (ISTA), Austria \and
	Technical University of Munich, Germany
}

\maketitle

\iftoggle{arxiv}{}{\vspace{-2em}}

\begin{abstract}
  Markov decision processes can be viewed as transformers of probability distributions. While this view is useful from a practical standpoint to reason about trajectories of distributions, basic reachability and safety problems are known to be computationally intractable (i.e., Skolem-hard) to solve in such models. Further, we show that even for simple examples of MDPs, strategies for safety objectives over distributions can require infinite memory and randomization.

  In light of this, we present a novel overapproximation approach to synthesize strategies in an MDP, such that a safety objective over the distributions is met. More precisely, we develop a new framework for template-based synthesis of certificates as affine distributional and inductive invariants for safety objectives in MDPs. We provide two algorithms within this framework. One can only synthesize memoryless strategies, but has relative completeness guarantees, while the other can synthesize general strategies. The runtime complexity of both algorithms is in PSPACE. We implement these algorithms and show that they can solve several non-trivial examples. 

\keywords{Markov decision processes \and invariant synthesis \and distribution transformers \and Skolem hardness}
\end{abstract}

\section{Introduction}
Markov decision processes (MDPs) are a classical model for probabilistic decision making systems. They extend the basic probabilistic model of Markov chains with non-determinism and are widely used across different domains and contexts. In the verification community, MDPs are often viewed through an automata-theoretic lens, as state transformers, with runs being sequences of states with certain probability for taking each run (see e.g.,~\cite{DBLP:books/daglib/0020348}). With this view, reachability probabilities can be computed using simple fixed point equations and model checking can be done over appropriately defined logics such as PCTL*. However, in several contexts such as modelling biochemical networks, queueing theory or probabilistic dynamical systems, it is more convenient to view MDPs as transformers of probability distributions over the states, and define objectives over these distributions~\cite{DBLP:conf/qest/ChadhaKVAK11,DBLP:conf/qest/KorthikantiVAK10,DBLP:conf/lics/AkshayGV18,DBLP:journals/tse/KwonA11,DBLP:journals/logcom/BeauquierRS06,DBLP:journals/jacm/AgrawalAGT15}. In this framework, we can, for instance, easily reason about properties such as the probability in a set of states always being above a given threshold or comparing the probability in two states at some future time point. More concretely, in a chemical reaction network, we may require that the concentration of a particular complex is never above $10\%$. Such distribution-based properties cannot be expressed in PCTL* \cite{DBLP:journals/logcom/BeauquierRS06}, and thus several orthogonal logics have been defined~\cite{DBLP:journals/logcom/BeauquierRS06,DBLP:conf/qest/KorthikantiVAK10,DBLP:journals/jacm/AgrawalAGT15} that reason about distributions.

Unfortunately, and perhaps surprisingly, when we view them as distribution transformers even the simplest reachability and safety problems with respect to probability distributions over states remain unsolved. The reason for this is a number-theoretical hardness result that lies at the core of these questions. In~\cite{DBLP:journals/ipl/AkshayAOW15}, it is shown that even with just Markov chains, reachability is as hard as the so-called \textsc{Skolem} problem, and safety is as hard as the \textsc{Positivity} problem~\cite{DBLP:conf/soda/OuaknineW14,DBLP:journals/siglog/OuaknineW15}, the decidability of both of which are long-standing open problems in linear recurrence sequences. Moreover, synthesizing strategies that resolve the non-determinism in MDPs to achieve an objective (whether reachability or safety) is further complicated by the issue of how much memory can be allowed for the strategy. As we show in Section~\ref{sec:problemandexamples}, even for very simple examples, strategies for safety can require infinite memory as well as randomization.  

In light of these difficulties, what can one do to tackle these problems \emph{in theory and in practice}? In this paper, we take an over-approximation route to approach these questions, not only to check existence of strategies for safety but also synthesize them. Inspired by the success of invariant synthesis in program verification, our goal is to develop a novel invariant-synthesis based approach towards strategy synthesis in MDPs, viewed as transformers of distributions.  In this paper, we restrict our attention to a class of safety objectives on MDPs, which are already general enough to capture several interesting and natural problems on MDPs. 

\paragraph*{Our contributions} are the following:
\begin{compactenum}
\item We define the notion of {\em inductive distributional invariants} for safety in MDPs. These are formalized as sets of probability distributions over states of the MDP, that (i) contain all possible distributions reachable from the initial distribution, under all  strategies of an MDP, and (ii) are closed under taking the next step. 
\item We show that such invariants provide {\em sound and complete certificates} for proving safety objectives in MDPs. In doing so, we formalize the link between strategies and distributional invariants in MDPs. This by itself does not help us get effective algorithms in light of the hardness results above. Hence we then focus on synthesizing invariants of a particular {\em shape}. 
\item We develop two algorithms for automated synthesis of {\em affine} inductive distributional invariants that prove safety in MDPs, and {\em at the same time}, synthesize the associated strategies.
  \begin{itemize}
  \item The first algorithm is restricted to synthesizing memoryless strategies but is {\em relatively complete}, i.e., whenever a memoryless strategy and an affine inductive distributional invariant that witness safety exist, we are guaranteed to find them.
  \item The second algorithm can synthesize general strategies as well as memoryless strategies, but is incomplete in general.
  \end{itemize}
  In both cases, we employ a template-based synthesis approach and reduce synthesis to the existential first-order theory of reals, which gives a PSPACE complexity upper bound. In the first case, this reduction depends on Farkas' lemma. In the second case, we need to use Handelman's theorem, a specialized result for strictly positive polynomials.
\item We implement our approaches and show that for several practical and non-trivial examples, affine invariants suffice. Further, we demonstrate that our prototype tool can synthesize these invariants as well as strategies associated with them. 
\end{compactenum}
Finally, we discuss the generalization of our approach from affine to polynomial invariants and some variants that our approach can handle. 


\subsection{Related Work}

\paragraph*{Distribution-based safety analysis in MDPs.} The problem of checking distribution-based safety objectives for MDPs was defined in~\cite{DBLP:conf/lics/AkshayGV18} but a solution was provided only in the {\em uninitialized} setting, where the initial distribution is not given and also under the assumption that the target set is closed and bounded. In contrast, we tackle both initialized and uninitialized settings, our target sets are general affine sets and we focus on actually synthesizing strategies not just proving existence.

\paragraph*{Template-based program analysis.} Template-based synthesis via the means of linear/polynomial constraint solving is a standard approach in program analysis to synthesizing certificates for proving properties of programs. Many of these methods utilize Farkas' lemma or Handelman's theorem to automate the synthesis of program invariants~\cite{ColonSS03,DBLP:conf/pldi/Chatterjee0GG20}, termination proofs~\cite{ColonS01,PodelskiR04,BradleyMS05,AliasDFG10,ChatterjeeG0Z21}, reachability proofs~\cite{DBLP:conf/pldi/AsadiC0GM21} or cost bounds~\cite{0002AH12,Carbonneaux0S15,ZikelicCBR22}. The works~\cite{SriramCAV,CNZ17,DBLP:conf/cav/ChatterjeeGMZ22,WFGCQS19,takisaka2021ranking,ChatterjeeGNZZ21,CFG16,AgrawalC018,CFNH16:prob-termination} utilize Farkas' lemma or Handelman's theorem to synthesize certificates for these properties in probabilistic programs. While our algorithms build on the ideas from the works on template-based inductive invariant synthesis in programs~\cite{ColonSS03,DBLP:conf/pldi/Chatterjee0GG20}, the key novelty of our algorithms is that they synthesize a fundamentally different kind of invariants, i.e.~{\em distributional invariants} in MDPs. In contrast, the existing works on (probabilistic) program analysis synthesize {\em state} invariants. Furthermore, our algorithms synthesize distributional invariants {\em together} with MDP strategies. While it is common in controller synthesis to synthesize an MDP strategy for a {\em state} invariant, we are not aware of any previous work that uses template-based synthesis methods to compute MDP strategies for a {\em distributional} invariant. 

\paragraph*{Other approaches to invariant synthesis in programs.} Alternative approaches to invariant synthesis in programs have also been considered, for instance via abstract interpretation~\cite{CousotC77,CousotCFMMMR05,FeautrierG10,RodriguezCarbonellK07}, counterexample guided invariant synthesis (CEGIS)~\cite{0001LMN14,AlurBDF0JKMMRSSSSTU15,BatzCJKKM23}, recurrence analysis~\cite{FarzanK15,KincaidBBR17,KincaidCBR18} or learning~\cite{0001NMR16,SiDRNS18}. While some of these approaches can be more scalable than constraint solving-based methods, they typically do not provide relative completeness guarantees. An interesting direction of future work would be to explore whether these alternative approaches could be used for synthesizing distributional invariants together with MDP strategies more efficiently.

\paragraph*{Weakest pre-expectation calculus.} Expectation transformers and the weakest pre-expectation calculus generalize Dijkstra's weakest precondition calculus to the setting of probabilistic programs. Expectation transformers were introduced in the seminal work on probabilistic propositional dynamic logic (PPDL)~\cite{Kozen83} and were extended to the setting of probabilistic programs with non-determinism in~\cite{MorganMS96,McIverM05}. Weakest pre-expectation calculus for reasoning about expected runtime of probabilistic programs was presented in~\cite{KaminskiKMO18}. Intuitively, given a function over probabilistic program outputs, the weakest pre-expectation calculus can be used to reason about the supremum or the infimum expected value of the function upon executing the probabilistic program, where the supremum and the infimum are taken over the set of all possible schedulers (i.e.\ strategies) used to resolve non-determinism. When the function is the indicator function of some output set of states, this yields the method for reasoning about the probability of reaching the set of states. Thus, weakest pre-expectation calculus allows reasoning about safety with respect to {\em sets of states}. In contrast, we are interested in reasoning about safety with respect to {\em sets of probability distribution over states}. Moreover, while the expressiveness of this calculus allows reasoning about very complex programs, its automation typically requires user input. In this work, we aim for a fully automated approach to checking distribution-based safety.

\section{Preliminaries}\label{sec:preliminaries}
In this section, we recall basics of probabilistic systems and set up our notation.
We assume familiarity with the central ideas of measure and probability theory, see \cite{billingsley2008probability} for a comprehensive overview.
%
We write $[n] := \{1, \dots, n\}$ to denote the set of all natural numbers from $1$ to $n$.
For any set $S$, we use $\overline{S}$ to denote its complement.
A \emph{probability distribution} on a countable set $X$ is a mapping $\distribution : X \to [0,1]$, such that $\sum_{x\in X} \distribution(x) = 1$.
Its \emph{support} is denoted by $\support(\distribution) = \{x \in X \mid \distribution(x) > 0\}$.
We write $\Distributions(X)$ to denote the set of all probability distributions on $X$.
An event happens \emph{almost surely} (a.s.) if it happens with probability $1$.
We assume that countable sets of states $\States$ are equipped with an arbitrary but fixed numbering. 

\subsection{Markov Systems}
A \emph{(discrete time) Markov chain (MC)} is a tuple $\MC = (\States, \transitions)$, where
	$\States$ is a finite set of \emph{states} and
	$\transitions : \States \to \Distributions(\States)$ a \emph{transition function}, assigning to each state a probability distribution over successor states.
A \emph{Markov decision process (MDP)} is a tuple $\MDP = (\States, \Actions, \transitions)$, where
	$\States$ is a finite set of \emph{states},
	$\Actions$ is a finite set of \emph{actions}, overloaded to yield for each state $s$ the set of \emph{available actions} $\Actions(s) \subseteq \Actions$, and
	$\transitions: \States \times \Actions \to \Distributions(\States)$ is a \emph{transition function} that for each state $s$ and (available) action $a \in \Actions(s)$ yields a probability distribution over successor states.
For readability, we write $\transitions(s, s')$ and $\transitions(s, a, s')$ instead of $\transitions(s)(s')$ and $\transitions(s, a)(s')$, respectively.
By abuse of notation, we redefine $\States \times \Actions := \{(s, a) \mid s \in \States \land a \in \Actions(s)\}$ to refer to the set of state-action pairs.
See \cref{fig:running} for an example MDP.
This MDP is our running example and we refer to it throughout this work to point out some of the peculiarities.

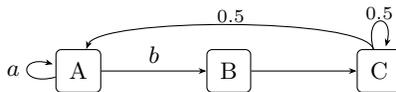
\begin{figure}[t]
	\centering
	\begin{tikzpicture}[auto]
		\node[state] at (0, 0) (A) {A};
		\node[state] at (2, 0) (B) {B};
		\node[state] at (4, 0) (C) {C};

		\draw[directedge]
			(A) edge[loop left] node[action] {$a$} (A)
			(A) edge node[action] {$b$} (B)
			(B) edge (C)
		;
		\draw[probedge]
			(C) edge[loop above] node[prob] {$0.5$} (C)
			(C) edge[out=105,in=70,looseness=0.3,swap] node[prob] {$0.5$} (A)
		;
	\end{tikzpicture}
	\caption{
		Our running example MDP.
		It comprises three states $S = \{A, B, C\}$, depicted by rounded rectangles.
		In state $A$, there are two actions available, namely $a$ and $b$.
		We have $\transitions(A, a, A) = 1$ and $\transitions(A, b, B) = 1$, indicated by arrows.
		States $B$ and $C$ have only one available action each, thus we omit explicitly labelling them.
	} \label{fig:running}
\end{figure}

An \emph{infinite path} in an MC is an infinite sequence $\infinitepath = s_1 s_2 \cdots \in \States^\omega$, such that for every $i \in \Naturals$ we have $\transitions(s_i, s_{i+1}) > 0$.
A \emph{finite path} $\finitepath$ is a finite prefix of an infinite path.
Analogously, infinite paths in MDP are infinite sequences $\infinitepath = s_1 a_1 s_2 a_2 \cdots \in (\States \times \Actions)^\omega$ such that $a_i \in \Actions(s_i)$ and $\transitions(s_i,a_i, s_{i+1}) > 0$ for every $i \in \Naturals$, and finite paths are finite prefixes thereof.
We use $\infinitepath_i$ and $\finitepath_i$ to refer to the $i$-th state in the given (in)finite path, and $\IPaths<M>$ and $\FPaths<M>$ for the set of all (in)finite paths of a system $M$.

\paragraph*{Semantics.}
A Markov chain evolves by repeatedly applying the probabilistic transition function in each step.
For example, if we start in state $s_1$, we obtain the next state $s_2$ by drawing a random state according to the probability distribution $\transitions(s_1)$.
Repeating this ad infinitum produces a random infinite path.
Indeed, together with an initial state $s$, a Markov chain $\MC$ induces a unique probability measure $\ProbabilityMC<\MC, s>$ over the (uncountable) set of infinite paths \cite{DBLP:books/daglib/0020348}.

This reasoning can be lifted to distributions over states, as follows.
Suppose we begin in $\distribution_0 = \{s_1 \mapsto 0.5, s_2 \mapsto 0.5\}$, meaning that initially we are in state $s_1$ or $s_2$ with probability $0.5$ each.
Then, $\distribution_1(s') = \distribution_0(s_1) \cdot \transitions(s_1, s') + \distribution_0(s_2) \cdot \transitions(s_2, s')$, i.e.\ the probability to be in a state $s'$ in the next step is $0.5$ times the probability of moving from $s_1$ and $s_2$ there, respectively.
For an initial distribution, we likewise obtain a probability distribution over infinite paths by setting $\ProbabilityMC<\MC, \distribution_0>[S] := \sum_{s \in \States} \distribution_0(s) \cdot \ProbabilityMC<\MC, s>[S]$ for measurable $S \subseteq \IPaths<\MC>$.

In contrast to Markov chains, MDPs also feature non-determinism, which needs be resolved in order to obtain probabilistic behaviour.
This is achieved by \emph{(path) strategies}, recipes to resolve non-determinism.
Formally, a strategy on an MDP classically is defined as a function $\strategy : \FPaths<\MDP> \to \Distributions(\Actions)$, which given a finite path $\finitepath = s_0 a_0 s_1 a_1 \dots s_n$ yields a probability distribution $\strategy(\finitepath) \in \Distributions(\Actions(s_n))$ on the actions to be taken next.
We write $\Strategies$ to denote the set of all strategies.
Fixing any strategy $\strategy$ induces a Markov chain $\MDP^\strategy = (\FPaths<\MDP>, \transitions^\strategy)$, where for a state $\finitepath = s_0 a_0 \dots s_n \in \FPaths<\MDP>$ the successor distribution is defined as $\transitions^\strategy(\finitepath, \finitepath a_{n+1} s_{n+1}) = \strategy(\finitepath, a_{n+1}) \cdot \transitions(s_n, a_{n+1}, s_{n+1})$.
(Note that the state space of this Markov chain in general is countably infinite.)
Consequently, for each strategy $\strategy$ and initial distribution $\distribution_0$ we also obtain a unique probability measure $\ProbabilityMC<\MDP^\strategy, \distribution_0>$ on the infinite paths of $\MDP$.
(Technically, the MC $\MDP^\strategy$ induces a probability measure over paths in $\MDP^\strategy$, i.e.\ paths where each element is a finite path of $\MDP$, however this can be directly projected to a measure over $\IPaths<\MDP>$.)

A \emph{one-step strategy} (also known as \emph{memoryless} or \emph{positional} strategy) corresponds to a fixed choice in each state, independent of the history, i.e.\ a mapping $\strategy : \States \to \Distributions(\Actions)$.
Fixing such a strategy induces a finite state Markov chain $\MDP^\strategy = (\States, \transitions^\strategy)$, where $\transitions^\strategy(s, s') = \sum_{a \in \Actions(s)} \strategy(s)(a) \cdot \transitions(s, a, s')$.
We write $\StrategiesOS$ for the set of all one-step strategies.

A sequence of one-step strategies $(\strategy_i) \in \StrategiesOS^\omega$ induces a general strategy which in each step $i$ and state $s$ chooses $\strategy_i(s)$.
Observe that aside from the state, such a strategy only depends on the current step, also called \emph{Markov strategy}.

%
\subsection{MDPs as Distribution Transformers}
Probabilistic systems typically are viewed as \enquote{random generators} for paths, and we consequently investigate the (expected) behaviour of a generated path, i.e.\ path properties.
However, in this work we follow a different view, and treat systems as \emph{transformers of distributions}.
Formally, fix a Markov chain $\MC$.
For a given initial distribution $\distribution_0$, we can define the distribution at step $i$ by $\distribution_i(s) = \ProbabilityMC<\distribution_0>[\{\infinitepath \in \IPaths<\MC> \mid \infinitepath_i = s\}]$.
We write $\distribution_i = \MC(\distribution_0, i)$ for the $i$-th distribution and $\distribution_1 = \MC(\distribution_0)$ for the \enquote{one-step} application of this transformation.
Likewise, we obtain the same notion for an MDP $\MDP$ combined with a strategy $\strategy$, and write $\distribution_i = \MDP^\strategy(\distribution_0, i)$, $\distribution_1 = \MDP^\strategy(\distribution_0)$.
In summary, for a given initial distribution, a Markov chain induces a unique stream of distributions, and an MDP provides one for each strategy.

This naturally invites questions related to this induced stream of distributions.
In their path interpretation, queries such as \emph{reachability} or \emph{safety}, i.e.\ asking the probability of reaching or avoiding a set of states, allow for simple, polynomial time solutions \cite{DBLP:books/wi/Puterman94,DBLP:books/daglib/0020348}.
However, the corresponding notions already are surprisingly difficult in the space of distributions.
Thus, we restrict to the \emph{safety problem}, which we introduce in the following.
Intuitively, given a \emph{safe set} of distributions over states $H \subseteq \Distributions(\States)$, we are interested in deciding whether the MDP can be controlled such that the stream of distributions always remains inside $H$.

\section{Problem Statement and Examples}\label{sec:problemandexamples}
Let $\MDP=(\States,\Actions,\transitions)$ be an MDP and $H \subseteq \Distributions(\States)$ be a safe set.
A distribution $\distribution_0$ is called \emph{$H$-safe under $\strategy$} if $\MDP^\strategy(\distribution_0, i) \in H$ for all $i \geq 0$, and \emph{$H$-safe} if there exists a strategy under which $\distribution_0$ is safe.
We mention two variants of the resulting decision problem as defined in~\cite{DBLP:conf/lics/AkshayGV18}:
\begin{compactitem}
	\item Initialized safety:
	Given an initial probability distribution $\distribution_0$ and safe set $H$, decide whether $\distribution_0$ is $H$-safe.
	
	\item Uninitialized safety:
	Given a safe set $H$, decide whether there exists a distribution $\mu$ which is $H$-safe.
\end{compactitem}
Note that we have discussed neither the shape nor the representation of $H$, which naturally plays an important role for decidability and complexity.

One may be tempted to think that the initialized variant is simpler, as more input is given.
However, this problem is known to be \emph{\textsc{Positivity}-hard}\footnote{
Intuitively, the \textsc{Positivity} problem asks for a given rational (or integer or real) matrix $M$, whether $(M^n)_{1,1}>0$ for all $n$ \cite{DBLP:conf/rp/OuaknineW12}. This problem (and its many variants) has been the subject of intense research over the last 10-15 years, see e.g.\ \cite{DBLP:conf/soda/OuaknineW14}.
Yet, quite surprisingly, it still remains open in its full generality.
} already for simple cases and already when $H$ is defined in terms of rational constants!
\begin{theorem}[\cite{DBLP:journals/ipl/AkshayAOW15}] \label{stm:skolem_hard}
	The initialized safety problem for Markov chains and $H$ given as linear inequality constraint ($H = \{\distribution \mid \distribution(s) \leq r, s\in \States, r\in \mathbb{Q}\cap [0,1]\}$), is \textsc{Positivity}-hard.
\end{theorem}
\begin{proof}
	In \cite[Corollary~4]{DBLP:journals/ipl/AkshayAOW15}, the authors show that the inequality version of the Markov reachability problem, i.e.\ deciding whether there exists an $i$ such that $\distribution_i(s) > r$ for a given rational $r$, is \textsc{Positivity}-hard.
	The result follows by observing that safety is the negation of reachability. \qed
\end{proof}
Thus, finding a decision procedure for this problem is unlikely, since it would answer several fundamental questions of number theory, see e.g.\ \cite{DBLP:conf/birthday/KarimovKO022,DBLP:journals/siglog/OuaknineW15,DBLP:conf/soda/OuaknineW14}.
In contrast, the uninitialized problem is known to be decidable for safe sets $H$ given as closed, convex polytopes (see \cite{DBLP:conf/lics/AkshayGV18} for details and \cite{DBLP:journals/jacm/AgrawalAGT15} for a different approach specific to Markov chains).
In a nutshell, we can restrict to the potential fixpoints of $\MDP$, i.e.\ all distributions $\distribution$ such that $\distribution = \MDP^\strategy(\distribution, i)$ for some strategy $\strategy$.
It turns out that this set of distributions is a polytope and the problem -- glossing over subtleties -- reduces to checking whether the intersection of $H$ with this polytope is non-empty.
However, we note that the solution of \cite{DBLP:conf/lics/AkshayGV18} does not yield the witness strategy.
In the following, we thus primarily focus on the initialized question.
In \cref{sec:extensions}, we then show how our approach, which also synthesizes a witness strategy, is directly applicable to the uninitialized case. 

In light of the daunting hardness results for the general initialized problem, we restrict to \emph{affine linear safe sets}, i.e.\ $H$ which are specified by a finite set of affine linear inequalities.
Formally, these sets are of the form $H = \{\distribution \in \Distributions(\States) \mid \bigwedge_{j=1}^N (c_0^j + \sum_{i=1}^n c_i^j \cdot \distribution(s_i)) \geq 0\}$, where $S=\{s_1,\dots,s_n\}$, $c_i^j$ are real-valued constants and $N$ is the number of affine linear inequalities that define $H$.
Our problem formally is given by the following query.
\begin{framed}
	\noindent \textbf{Problem Statement} Given an MDP $\MDP$, initial distribution $\distribution_0$, and affine linear safe set $H$, (i)~decide whether $\distribution_0$ is $H$-safe, and (ii)~if yes, then synthesize a strategy for $\MDP$ which ensures safety.
\end{framed}
Note that the problem strictly subsumes the special case when $H$ is defined in terms of rational constants, and our approach aims to solve both problems. Also, note that \cref{stm:skolem_hard} still applies, i.e.\ this \enquote{simplified} problem is \textsc{Positivity}-hard, too.
We thus aim for a sound and \emph{relatively complete} approach.
Intuitively, this means that we restrict our search to a sub-space of possible solutions and within this space provide a complete answer.
To give an intuition for the required reasoning, we provide an example safety query together with a manual proof.
\begin{example} \label{ex:first_simple_case}
	Consider our running example from \cref{fig:running}.
	Suppose the initial distribution is $\distribution_0 = \{A \mapsto \frac{1}{3}, B \mapsto \frac{1}{3}, C \mapsto \frac{1}{3}\}$ and (affine linear) $H = \{\distribution \mid \distribution(C) \geq \frac{1}{4}\}$.
	This safety query is satisfiable, by, e.g., choosing action $b$, as we show in the following.
	First, observe that the $i+1$-th distribution is $\distribution_{i+1}(A) = \frac{1}{2} \cdot \distribution_i(C)$, $\distribution_{i+1}(B) = \distribution_i(A)$, and $\distribution_{i+1}(C) = \distribution_i(B) + \frac{1}{2} \distribution_i(C)$.
	Thus, we cannot directly prove by induction that $\distribution_i(C) \geq \frac{1}{4}$, we also need some information about $\distribution_i(B)$ or $\distribution_i(A)$ to exclude, e.g., $\distribution_i = \{A \mapsto \frac{3}{4}, C \mapsto \frac{1}{4}\}$, where $\distribution_{i+1}$ would violate the safety constraint.
	We invite the interested reader to try to prove that $\distribution_0$ is indeed $H$-safe under the given strategy to appreciate the subtleties.

	We proceed by proving that $\distribution_i(C) \geq \frac{1}{4}$ and additionally $\distribution_i(A) \leq \distribution_i(C)$ by induction.
	The base case follows immediately, thus suppose that $\distribution_i$ satisfies these constraints.
	For $\distribution_{i+1}(A) \leq \distribution_{i+1}(C)$ observe that $\distribution_{i+1}(A) = \frac{1}{2} \distribution_i(C)$ and $\distribution_{i+1}(C) = \frac{1}{2} \distribution_i(C) + \distribution_i(B)$.
	Since $\distribution_i(B) \geq 0$, the claim follows.
	To prove $\distribution_{i+1}(C) \geq \frac{1}{4}$ 
	observe that $\distribution_i(A) \leq \frac{1}{2}$ since $\distribution_i(A) \leq \distribution_i(C)$ by induction hypothesis and distributions sum up to $1$.
	Moreover, $\distribution_{i+1}(C) = \distribution_i(B) + \frac{1}{2} \distribution_i(C) = \frac{1}{2} \distribution_i(B) + \frac{1}{2} - \frac{1}{2} \distribution_i(A)$ by again inserting the fact that distributions sum up to $1$.
	Then, $\distribution_{i+1}(C) = \frac{1}{2} - \frac{1}{2} \distribution_i(A) + \frac{1}{2} \distribution_i(B) \geq \frac{1}{2} - \frac{1}{2} \distribution_i(A) \geq \frac{1}{2} - \frac{1}{4} \geq \frac{1}{4}$. \qee
\end{example}
Thus, already for rather simple examples the reasoning is non-trivial.
To further complicate things, the structure of strategies can also be surprisingly complex:
\begin{example}
	Again consider our running example from \cref{fig:running} with initial distribution $\distribution_0 = \{A \mapsto \frac{3}{4}, B \mapsto \frac{1}{4}\}$ and safe set $H = \{\distribution \mid \distribution(B) = \frac{1}{4}\}$.
	This safety condition is indeed satisfiable, however the (unique) optimal strategy requires both infinite memory as well as randomization with arbitrarily small fractions!
	In step $1$, we require choosing $a$ with $\frac{2}{3}$ and $b$ with $\frac{1}{3}$ to satisfy the safety constraint in the second step, getting $\distribution_1 = \{A \mapsto \frac{1}{2}, B \mapsto \frac{1}{4}, C \mapsto \frac{1}{4}\}$.
	For step $2$, we require choosing both $a$ and $b$ with probability $\frac{1}{2}$ each, yielding $\distribution_2 = \{A \mapsto \frac{3}{8}, B \mapsto \frac{1}{4}, C \mapsto \frac{3}{8}\}$.
	Continuing this strategy, we obtain at step $i$ that $\distribution_{i} = \{A \mapsto \frac{1}{4} + \frac{1}{2^{i+1}}, B \mapsto \frac{1}{4}, C \mapsto \frac{1}{2} - \frac{1}{2^{i+1}}\}$ and action $a$ is chosen with probability $1 / (2^{i-1} + 1)$, converging to $1$. \qee
\end{example}
In the following, we provide two algorithms that handle both examples.
Our first algorithm focusses on memoryless strategies, the second considers a certain type of infinite memory strategies.
Essentially, the underlying idea is to automatically synthesize a strategy together with such inductive proofs of safety.

\section{Proving Safety by Invariants}\label{sec:invariantssec}
We now discuss our principled idea of proving safety by means of (inductive) invariants, taking inspiration from research on safety analysis in programs~\cite{ColonSS03,DBLP:conf/pldi/Chatterjee0GG20}.
We first show that considering strategies which are purely based on the current distribution over states are sufficient.
Then, we show that inductive invariants are a \emph{sound and complete} certificate for safety.
Together, we obtain that an initial distribution is $H$-safe \emph{if and only if} there exists an invariant set $I$ and distribution strategy $\strategy$ such that (i)~the initial distribution is contained in $I$, (ii)~$I$ is a subset of the safe set $H$, and (iii)~$I$ is inductive under $\strategy$, i.e.\ if $\distribution \in I$ then $\MDP^\strategy(\distribution) \in I$.
In the following section, we then show how we search for invariants and distribution strategies \emph{of a particular shape}.

\subsection{Distribution Strategies}
We show that \emph{distribution strategies} $\strategy : \Distributions(\States) \to \StrategiesOS$, yielding for each distribution over states a one-step strategy to take next, are sufficient for the problem at hand.
More formally, we want to show that an $H$-safe distribution strategy exists if and only if there exists any $H$-safe strategy.

First, observe that distribution strategies are a special case of regular path strategies.
In particular, for any given initial distribution, we obtain a uniquely determined stream of distributions as $\distribution_{i+1} = \MDP^{\strategy(\distribution_i)}(\distribution_i)$, i.e.\ the distribution $\distribution_{i+1}$ is obtained by applying the one-step strategy $\strategy(\distribution_i)$ to $\distribution_i$.
In turn, this lets us define the Markov strategy $\hat{\strategy}_i(s) = \strategy(\distribution_i)(s)$.
For simplicity, we identify distribution strategies with their induced path strategy.

Next, we argue that restricting to distribution strategies is sufficient.
\begin{theorem} \label{stm:dist_strat_is_enough}
	An initial distribution $\distribution_0$ is $H$-safe if and only if there exists a distribution strategy $\strategy$ such that $\distribution_0$ is $H$-safe under $\strategy$.
\end{theorem}
\iftoggle{arxiv}{
\begin{proof}
	The backward direction follows immediately.

	For the forward direction, suppose that $\distribution_0$ is $H$-safe and let $\strategy$ be a witness thereof.
	We consider the stream of induced distributions $\distribution_i = \MDP^\strategy(\distribution_0, i)$ together with the distribution of played actions in each step, i.e.\ $p_{s, a} = \ProbabilityMC<\MDP^\strategy, \distribution_0>[\mathsf{act}_i(\infinitepath) = a]$, where $\mathsf{act}_i : \IPaths<\MDP> \to \Actions$ yields the $i$-th action of an infinite path.
	Note that $\distribution_i \in H$ for all $i$ by assumption.
	We define the distribution strategy $\hat{\strategy}$ as follows:
	For each occurring distribution $\distribution_i$, we set $\hat{\strategy}(\distribution_i) = \strategy_i$, where $\strategy_i(s)$ is arbitrary if $\distribution_i(s) = 0$ and $\strategy_i(s)(a) = p_{s, a} / \sum_{a \in \Actions(s)} p_{s, a}$ otherwise.
	Note that if $\distribution_i(s) > 0$, we necessarily have that $\sum_{a \in \Actions(s)} p_{s, a} > 0$, since at step $i$ we are in state $s$ with non-negative probability and thus some action in $s$ is played by the original strategy $\strategy$.
	Clearly, applying $\MDP^{\strategy_i}(\distribution_i) = \distribution_{i+1}$.

	To ensure that $\hat{\strategy}$ is well-defined, we need to consider a special case, namely when a distribution appears several times, i.e.\ if $\distribution_i = \distribution_j$ for $i \neq j$.
	Suppose that $i$ is the first time a distribution appears again under $\strategy$, i.e.\ $\distribution_i = \distribution_j$ for $j < i$ and $\distribution_{j} \neq \distribution_{j'}$ for all $j < j' < i$.
	Then, we only define $\strategy(\distribution)$ for all distributions before step $i$.
	As the distribution $\distribution_j$ re-appears in step $i$, we effectively closed a loop inside $H$ and we can simply keep re-applying the decision between step $j$ and $i$ to remain safe. \qed
\end{proof}
}{
\begin{proof}[Sketch]
	The full proof can be found in \cite[Sec.\ 4.1]{arxiv}.
	Intuitively, only the \enquote{distribution} behaviour of a strategy is relevant and we can sufficiently replicate the behaviour of any safe strategy by a distribution strategy. \qed
\end{proof}
}
In this way, each MDP corresponds to a (uncountably infinite) transition system $\TS<\MDP> = (\Distributions(\States), T)$ where $(\distribution, \distribution') \in T$ if there exists a one-step strategy $\strategy$ such that $\distribution' = \MDP^\strategy(\distribution)$.
Note that $\TS<\MDP>$ is a purely non-deterministic system, without any probabilistic behaviour.
So, our decision problem is equivalent to asking whether the induced transition system $\TS<\MDP>$ can be controlled in a safe way.
Note that $\TS<\MDP>$ is uncountably large and uncountably branching.

\subsection{Distributional Invariants for MDP Safety} \label{sec:invariants}
%
We now define distributional invariants in MDPs and show that they provide sound and complete certificates for proving initialized (and uninitialized) safety.

\paragraph*{Distributional Invariants in MDPs.}
Intuitively, a distributional invariant is a set of probability distributions over MDP states that contains all probability distributions that can arise from applying a strategy to an initial probability distribution, i.e.\ the complete stream $\distribution_i$.
Hence, similar to the safe set $H$, distributional invariants are also defined to be subsets of $\Distributions(\States)$.

\begin{definition}[Distributional Invariants]\label{def:inv}
	Let $\distribution_0 \in \Distributions(\States)$ be a probability distribution over $\States$ and $\strategy$ be a strategy in $\MDP$.
	A set $I \subseteq \Distributions(\States)$ is said to be a \emph{distributional invariant for $\distribution_0$ under $\strategy$} if the sequence of probability distributions induced by applying the strategy $\strategy$ to the initial probability distribution $\distribution_0$ is contained in $I$, i.e.\ if $\MDP^{\strategy}(\distribution_0, i) \in I$ for each $i \geq 0$.

	A distributional invariant $I$ is said to be \emph{inductive under $\strategy$}, if we furthermore have that $\MDP^{\strategy}(\distribution) \in I$ holds for any $\distribution \in I$, i.e.\ if $I$ is \enquote{closed} under application of $\MDP^\strategy$ to any probability distribution contained in $I$.
\end{definition}

\paragraph*{Soundness and Completeness for MDP Safety.}
The following theorem shows that, in order to solve the initialized (and uninitialized) safety problem, one can equivalently search for a distributional invariant that is fully contained in $H$.
Furthermore, it shows that one can without loss of generality restrict the search to inductive distributional invariants.

\begin{theorem}[Sound and Complete Certificate] \label{thm:inv}
	Let $\distribution_0 \in \Distributions(\States)$ be a probability distribution over $\States$, $\strategy$ be a strategy in $\MDP$, and $H \subseteq \Distributions(\States)$ be a safe set.
	Then $\distribution_0$ is $H$-safe under $\strategy$ if and only if there exists an inductive distributional invariant $I$ for $\distribution_0$ and $\strategy$ such that $I \subseteq H$.
\end{theorem}
\iftoggle{arxiv}{
\begin{proof}
	Suppose first that there exists an inductive distributional invariant $I$ for $\distribution_0$ and $\strategy$ such that $I\subseteq H$.
	Then, by the definition of distributional invariants, we know that $\MDP^{\strategy}(\distribution_0, i) \in I$ for each $i \geq 0$.
	Thus, as $I\subseteq H$, this implies that $\MDP^{\strategy}(\distribution_0, i) \in H$ for each $ i \geq 0$.
	Hence, $\strategy$ is $H$-safe from $\distribution_0$.
	
	For the opposite direction, suppose that $\strategy$ is $H$-safe from $\distribution_0$.
	We define a set $I \subseteq \Distributions(\States)$ to contain exactly those probability distributions that are induced by applying $\strategy$ to $\distribution_0$, i.e.\ $I = \Union_{i=0}^\infty \{\MDP^{\strategy}(\distribution_0, i)\}$.
	This clearly is an inductive distributional invariant as defined by \cref{def:inv}.
	On the other hand, since $\strategy$ is $H$-safe from $\distribution_0$, we have that this whole sequence is contained in $H$, hence $I\subseteq H$.
	This proves the claim. \qed
\end{proof}
}{The proof can be found in \cite[Sec.\ 4.2]{arxiv}.}

Thus, in order to solve the initialized safety problem for $\distribution_0$, it suffices to search for (i)~a strategy $\strategy$ and (ii)~an inductive distributional invariant $I$ for $\distribution_0$ and $\strategy$ such that $I\subseteq H$.
On the other hand, in order to solve the uninitialized safety problem, it suffices to search for (i)~an initial probability distribution $\distribution_0$, (ii)~strategy $\strategy$, and (iii)~an inductive distributional invariant $I$ for $\distribution_0$ and $\strategy$ such that $I\subseteq H$.
In the following, we provide a fully automated, sound and \emph{relatively} complete method of deciding the existence of such an invariant and strategy.

\section{Algorithms for Distributional Invariant Synthesis}\label{sec:algo}

We now present two algorithms for automated synthesis of strategies and inductive distributional invariants towards solving distribution safety problems in MDPs.
The two algorithms differ in the kind of strategies they consider and, as a consequence of differences in the involved expressions, also in their completeness guarantees.
For readability, we describe the algorithms in their basic form applied to the initialized variant of the safety problem and discuss further extensions in \cref{sec:extensions}.
In particular, our approach is also directly applicable to the uninitialized variant, as we describe there.

We say that an inductive distributional invariant is \emph{affine} if it can be specified in terms of (non-strict) affine inequalities, which we formalize below.
Both algorithms jointly synthesize a strategy and an affine inductive distributional invariant by employing a \emph{template-based synthesis} approach.
In particular, they fix symbolic templates for each object that needs to be synthesized, encode the defining properties of each object as constraints over unknown template variables, and solve the system of constraints by reduction to the existential first-order theory of the reals.

\newcommand{\tpl}[1]{\textcolor{gray}{#1}}
For example, a template for an affine linear constraint on distributions $\Distributions(\States)$ is given by $\mathsf{aff}(\distribution) = (\tpl{c_0} + \tpl{c_1} \cdot \distribution(s_1) + \dots + \tpl{c_n} \cdot \distribution(s_n) \geq 0)$.
Here, the variables $c_0$ to $c_n$, written in grey for emphasis, are the \emph{template variables}.
For fixed values of these variables the expression $\mathsf{aff}$ is a concrete affine linear predicate over distributions.
Thus, we can ask questions like \enquote{Do there exist values for $c_i$ such that for all distributions $\distribution$ we have that $\mathsf{aff}(\distribution)$ implies $\mathsf{aff}(\MDP^\strategy(\distribution))$?}.
This is a sentence in the theory of reals -- however with quantifier alternation.
As a next step, template-based synthesis approaches then employ various quantifier elimination techniques to convert such expressions into equisatisfiable sentences in, e.g., the existential theory of reals, which is decidable in PSPACE \cite{DBLP:conf/stoc/Canny88}.


\paragraph*{Difference between the Algorithms.}
Our two algorithms differ in their applicability and the kind of completeness guarantees that they provide.
In terms of applicability, the first algorithm only considers \emph{memoryless} strategies, while the second algorithm searches for \emph{distribution} strategies specified as fractions of affine linear expressions.
(We discuss an extension to rational functions in \cref{sec:extensions}.)
In terms of completeness guarantees, the first algorithm is \emph{(relatively) complete} in the sense that it is guaranteed to compute a memoryless strategy and an affine inductive distributional invariant that prove safety \emph{whenever they exist}.
In contrast, the second algorithm does not provide the same level of completeness.

\paragraph*{Notation.}
In what follows, we write ${\equiv}$ to denote (syntactic) equivalence of expressions, to distinguish from relational symbols used inside these expressions, such as \enquote{$=$}.
For example $\Phi(x) \equiv x = 0$ means that $\Phi(x)$ is the predicate $x = 0$.
Moreover, $(x_1,\dots,x_n)$ denotes a symbolic probability distribution over the state space $\States=(\state_1,\dots,\state_n)$, where $x_i$ is a symbolic variable that encodes the probability of the system being in $\state_i$.
We use boldface notation $\vec{x} = (x_1,\dots,x_n)$ to denote the vector of symbolic variables.
Thus, the above example would be written $\mathsf{aff}(\vec{x}) \equiv c_0 + c_1 \cdot x_1 + \dots + c_n \cdot x_n \geq 0$.
Since we often require vectors to represent a distribution, we write $\vec{x} \in \Distributions(\States)$ as abbreviation for the predicate ${\bigwedge}_{i=1}^n (0 \leq x_i \leq 1) \land ({\sum}_{i=1}^n x_i = 1)$.

\paragraph*{Algorithm Input and Assumptions.}
Both algorithms take as input an MDP $\MDP = (\States, \Actions, \transitions)$ with $\States = \{\state_1,\dots,\state_n\}$.
They also take as input a safe set $H \subseteq \Distributions(\States)$.
We assume that $H$ is specified by a boolean predicate over $n$ variables as a logical conjunction of $N_H \in \mathbb{N}_0$ \emph{affine} inequalities, and that it has the form
\begin{equation*}
	H(\vec{x}) \equiv (\vec{x} \in \Distributions(\States)) \land {\bigwedge}_{i=1}^{N_H} (h^i(\vec{x}) \geq 0),
\end{equation*}
where the first term imposes that $\vec{x}$ is a probability distribution over $\States$ and $h^i(\vec{x}) = h_0^i + h_1^i \cdot x_1 + \dots + h_n^i \cdot x_n$ is an affine expression over $\vec{x}$ with real-valued coefficients $h^i_j$ for each $i \in [N_H]$ and $j \in \{0, \dots, n\}$.
(Note that $h_j^i$ are not template variables but fixed values, given as input.)
Next, the algorithms take as input an initial probability distribution $\distribution_0 \in \Distributions(\States)$.
Finally, the algorithms also take as input technical parameters.
Intuitively, these describe the size of used \emph{symbolic templates}, explained later.
For the remainder of the section, fix an initialized safety problem, i.e.\ an $\MDP$, safe set $H$ of the required form, and an initial distribution $\distribution_0$.

\subsection{Synthesis of Affine Invariants and Memoryless Strategies}\label{sec:algoone}

We start by presenting our first algorithm, which synthesizes memoryless strategies and affine inductive distributional invariants.
We refer to this algorithm as \algomem{}.
The algorithm proceeds in the following four steps:
\begin{compactenum}
	\item \emph{Setting up Templates.} The algorithm fixes symbolic templates for the memoryless strategy $\strategy$ and the affine inductive distributional invariant $I$. 
	Note that the values of the symbolic template variables at this step are \emph{unknown} and are to be computed in subsequent steps.
	\item \emph{Constraint Collection.} The algorithm collects the constraints which encode that $\strategy$ is a (memoryless) strategy, that $I$ contains the initial probability distribution $\distribution_0$, that $I$ is an inductive distributional invariant with respect to $\strategy$ and $\distribution_0$, and that $I$ is contained within $H$.
	This step yields a system of affine constraints over symbolic template variables that contain universal and existential quantifiers.
	\item \emph{Quantifier Elimination.} The algorithm eliminates universal quantifiers from the above constraints to reduce it to a system of purely existentially quantified system of polynomial constraints over the symbolic template variables.
	Concretely, the first algorithm achieves this by application of \emph{Farkas' lemma}.
	\item \emph{Constraint Solving.} The algorithm solves the resulting system of constraints by using an off-the-shelf solver to compute concrete values for symbolic template variables specifying the strategy $\strategy$ and invariant $I$.
\end{compactenum}
We now describe each step in detail.

\paragraph*{Step~1: Setting up Templates.} The algorithm sets templates for $\strategy$ and $I$ as follows:
\begin{compactitem}
	\item Since this algorithm searches for memoryless strategies, the probability of taking an action $\action_j$ in state $\state_i$ is always the same, independent of the current distribution.
	Hence, our template for $\strategy$ consists of a symbolic template variable $p_{\state_i, \action_j}$ for each $\state_i \in \States$, $\action_j \in \Actions(s_i)$.
	We write $p_{\state_i, \circ} = (p_{\state_i, \action_1}, \dots, p_{\state_i, \action_m})$ to refer to the corresponding distribution in state $\state_i$.
	\item The template of $I$ is given by a boolean predicate specified by a conjunction of $N_I$ affine inequalities, where $N_I$ is the \emph{template size} and is an algorithm parameter.
	In particular, the template of $I$ looks as follows:
	\begin{equation*}
		I(\vec{x}) \equiv (\vec{x} \in \Distributions(\States)) \land {\bigwedge}_{i=1}^{N_I} (a^i_0 + a^i_1\cdot x_1 + \dots + a^i_n\cdot x_n \geq 0).
	\end{equation*}
	The first predicate enforces that $I$ only contains vectors that define probability distributions over $S$.
\end{compactitem}

\paragraph*{Step~2: Constraint Collection.}
We now collect the constraints over symbolic template variables which encode that $\strategy$ is a memoryless strategy, that $I$ contains the initial distribution $\distribution_0$, that $I$ is an inductive distributional invariant under $\strategy$, and that $I$ is contained in $H$.
\begin{compactitem}
	\item For $\strategy$ to be a strategy, we only need to ensure that each $p_{\state_i, \circ}$ is a probability distribution over the set of available actions at every state $\state_i$.
	Thus, we set
	\begin{equation*}
		\Phi_{\textrm{strat}} \equiv {\bigwedge}_{i=1}^n \left( p_{\state_i, \circ} \in \Distributions(\Actions(\state_i)) \right).
	\end{equation*}
	\item For $I$ to be a distributional invariant for $\strategy$ and $\distribution_0$ as well as to be inductive, it suffices to enforce that $I$ contains $\distribution_0$ and that $I$ is closed under application of $\strategy$.
	Thus, we collect two constraints:
	\begin{equation*}
		\begin{split}
			\Phi_{\textrm{initial}} &\equiv I(\distribution_0) \equiv {\bigwedge}_{i=1}^{N_I} (a^i_0 + a^i_1 \cdot \distribution_0^1 + \dots a^i_n\cdot \distribution_0^n \geq 0), \text{ and} \\
			\Phi_{\textrm{inductive}} &\equiv \left(\forall \vec{x} \in \mathbb{R}^n.\ I(\vec{x}) \Longrightarrow I(\mathrm{step}(\vec{x})) \right),
		\end{split}
	\end{equation*}
	where $\mathrm{step}(\vec{x})(x_i) = \sum_{\state_k \in \States, \action_j \in \Actions(s_k)} p_{\state_k, \action_j} \cdot \transitions(\state_k, \action_j, \state_i) \cdot x_j$ yields the distribution after applying one step of the strategy induced by $\Phi_{\textrm{strat}}$ to $\vec{x}$.
	\item For $I$ to be contained in $H$, we enforce the constraint:
	\begin{equation*}
		\Phi_{\textrm{safe}} \equiv \left(\forall \vec{x}\in\mathbb{R}^n.\ I(\vec{x}) \Longrightarrow H(\vec{x}) \right).
	\end{equation*}
\end{compactitem}

\paragraph*{Step~3: Quantifier Elimination.}
Constraints $\Phi_{\textrm{strat}}$ and $\Phi_{\textrm{initial}}$ are purely existentially quantified over symbolic template variables, thus we can solve them directly.
However, $\Phi_{\textrm{inductive}}$ and $\Phi_{\textrm{safe}}$ contain both universal and existential quantifiers, which are difficult to handle.
In what follows, we show how the algorithm translates these constraints into equisatisfiable \emph{purely existentially quantified} constraints.
In particular, our translation exploits the fact that both $\Phi_{\textrm{inductive}}$ and $\Phi_{\textrm{safe}}$ can, upon splitting the conjunctions on the right-hand side of implications into conjunctions of implications, be expressed as conjunctions of constraints of the form
\begin{equation*}
	\forall\vec{x}\in\mathbb{R}^n.\ (\textrm{affexp}_1(\vec{x}) \geq 0) \land \dots \land (\textrm{affexp}_N(\vec{x}) \geq 0) \Longrightarrow (\textrm{affexp}(\vec{x}) \geq 0).
\end{equation*}
Here, each $\textrm{affexp}_i(\vec{x})$ and $\textrm{affexp}(\vec{x})$ is an affine expression over $\vec{x}$ whose affine coefficients are either concrete real values or symbolic template variables.

In particular, we use Farkas' lemma~\cite{farkas1902theorie} to remove universal quantification and translate the constraint into an equisatisfiable existentially quantified system of constraints over the symbolic template variables, as well as fresh auxiliary variables that are introduced by the translation.
For completeness, we briefly recall (a strengthened and adapted version of) Farkas' lemma.
\begin{lemma}[\cite{farkas1902theorie,DBLP:books/daglib/0016926}] \label{thm:strengthenedfarkas}
	Let $\mathcal{X} = \{x_1,\dots,x_n\}$ be a finite set of real-valued variables, and consider the following system of $N\in\mathbb{N}$ affine inequalities over $\mathcal{X}$:
	\begin{equation*}
		\Phi : \begin{cases}
			c^1_0 + c^1_1\cdot x_1 + \dots + c^1_n \cdot x_n \geq 0 \\
			\qquad \qquad \qquad \vdots \\
			c^N_0 + c^N_1\cdot x_1 + \dots + c^N_n \cdot x_n \geq 0 \\
		\end{cases}.
	\end{equation*}
	Suppose that $\Phi$ is satisfiable.
	Then $\Phi$ entails an affine inequality $\phi \equiv c_0 + c_1 \cdot x_1 + \dots + c_n \cdot x_n$, i.e.\ $\Phi \Longrightarrow \phi$, if and only if $\phi$ can be written as a \emph{non-negative} linear combination of affine inequalities in $\Phi$, i.e.\ if and only if there exist $y_1,\dots,y_n\geq 0$ such that $c_1 = \sum_{j=1}^N y_j \cdot c^j_1$, \dots, $c_n = \sum_{j=1}^N y_j \cdot c^j_n$.
\end{lemma}

Note that, for any implication appearing in $\Phi_{\textrm{inductive}}$ and $\Phi_{\textrm{safe}}$, the system of constraints on the left-hand side is simply $I(\vec{x})$, and the satisfiability of $I(\vec{x})$ is enforced by $\Phi_{\textrm{initial}}$.
Hence, we may apply Farkas lemma to translate each constraint with universal quantification into an equivalent purely existentially quantified constraint.
In particular, for any constraint of the form
\begin{equation*}
	\forall\vec{x}\in\mathbb{R}^n.\, (\textrm{affexp}_1(\vec{x}) \geq 0) \land \dots \land (\textrm{affexp}_N(\vec{x}) \geq 0) \Longrightarrow (\textrm{affexp}(\vec{x}) \geq 0),
\end{equation*}
we introduce fresh template variables $y_1,\dots,y_N$ and translate it into the system of purely existentially quantified constraints
\begin{equation*}
	(y_1 \geq 0) \land \dots \land (y_N \geq 0) \land (\textrm{affexp}(\vec{x}) \equiv_{F} y_1 \cdot \textrm{affexp}_1(\vec{x}) + \dots + y_N \cdot \textrm{affexp}_N(\vec{x})).
\end{equation*}
Here, we use $\textrm{affexp}(\vec{x}) \equiv_{F} y_1 \cdot \textrm{affexp}_1(\vec{x}) + \dots + y_N \cdot \textrm{affexp}_N(\vec{x})$ to denote the set of $n+1$ equalities over the symbolic template variable and $y_1,\dots,y_N$ which equate the constant coefficients as well as the linear coefficients of each $x_i$ on two sides of the equivalence, i.e.\ exactly those equalities which we obtain from applying Farkas' lemma.
We highlight that the expressions $\textrm{affexp}$ are only affine linear for \emph{fixed} existentially quantified variables, i.e.\ they are in general quadratic.

\paragraph*{Step~4: Constraint Solving.}
Finally, we feed the resulting system of existentially quantified polynomial constraints over the symbolic template variables as well as the auxiliary variables introduced by applying Farkas' lemma to an off-the-shelf constraint solver.
If the solver outputs a solution, we conclude that the computed invariant $I$ is an inductive distributional invariant for the strategy $\strategy$ and initial distribution $\distribution_0$, and that $I$ is contained in $H$.
Therefore, by \cref{thm:inv}, we conclude that $\distribution_0$ is $H$-safe under $\strategy$.
\begin{theorem}\label{thm:firstalgo}
	\emph{Soundness}: Suppose \algomem{} returns a memoryless strategy $\strategy$ and an affine inductive distributional invariant $I$.
	Then, $\distribution_0$ is $H$-safe under $\strategy$.

	\emph{Completeness}: If there exist a memoryless strategy $\strategy$ and an affine inductive distributional invariant $I$ such that $I\subseteq H$ and $\distribution_0$ is $H$-safe under $\strategy$, then there exists a minimal value of the template size $N_I \in \mathbb{N}$ such that $\strategy$ and $I$ are produced by \algomem{}.

	\emph{Complexity}: The runtime of \algomem{} is in PSPACE in the size of the MDP, the encoding of the safe set $H$ and the template size parameter $N_I \in \mathbb{N}$. 
\end{theorem}
\iftoggle{arxiv}{
\begin{proof}
	
	To prove soundness, i.e.\ the first part of the theorem claim, suppose that  \algomem{} returns a memoryless strategy $\strategy$ and an affine inductive distributional invariant $I$. This means that, in Step~4, the algorithm computed a solution $(\strategy, I, y_1, \dots, y_N)$ to the system of constraints constructed in Step~3, with $y_1,\dots, y_N \geq 0$. Here we slightly abuse the notation and use $\strategy$ and $I$ to also denote the values of template variables that specify $\strategy$ and $I$. By Lemma~\ref{thm:strengthenedfarkas}, this then means that $(\strategy, I)$ is a solution to the system of constraints constructed in Step~2. But constraints in Step~2 encode that $\strategy$ is a memoryless strategy, $I$ is a distributional invariant for $\strategy$ and $\distribution_0$ and that $I$ is contained in $H$. Hence, by Theorem~\ref{thm:inv} it follows that $\distribution_0$ is $H$-safe under $\strategy$, which proves the first part of the theorem claim.
	
	To prove completeness, i.e.\ the second part of the theorem claim, suppose that there exist a memoryless strategy $\strategy$ and an affine inductive distributional invariant $I$ such that $I\subseteq H$ and $\distribution_0$ is $H$-safe under $\strategy$. We need to show that there exist $y_1,\dots, y_N \geq 0$ and the minimal template size $N_I$ such that $(\strategy, I, y_1, \dots, y_N)$ is a solution to the system of constraints in Step~3 for the template size $N_I$, which the algorithm can thus compute in Step~$4$. To prove this, define $N_I$ as the number of affine inequalities appearing in the specification of $I$. By assumptions in the second part of the theorem claim, it follows that $(\strategy, I)$ satisfy all constraints and therefore present a solution to the system of constraints constructed in Step~2 for the template size $N_I$. Then, by Lemma~\ref{thm:strengthenedfarkas}, we have that there exist $y_1,\dots, y_N \geq 0$ such that $(\strategy, I, y_1, \dots, y_N)$ is a solution to the system of constraints in Step~3 for the template size $N_I$. This proves the second part of the theorem claim.


	For the runtime complexity, observe that the first three steps of the algorithm all have polynomial runtime and yield a system of constraints which is polynomial in the size of the MDP, the encoding of the safe set $H$ and the template size parameter $N_I\in \mathbb{N}$.
	Thus, the resulting query is a sentence in the existential first-order theory of the reals, which can be solved in PSPACE. \qed
\end{proof}
}{The proof can be found in \cite[Sec.\ 5.1]{arxiv}.}
We comment on the PSPACE upper bound on the complexity of \algomem{}.
The upper bound holds since the application of Farkas' lemma reduces synthesis to solving a sentence in the existential first-order theory of the reals and since the size of the sentence is polynomial in the sizes of the MDP, the encoding of the safe set $H$ and the invariant template size $N_i$.
However, it is unclear whether the resulting constraints could be solved more efficiently, and the best known upper bound on the time complexity of algorithms for template-based affine inductive invariant synthesis in programs is also PSPACE~\cite{ColonSS03,DBLP:conf/pldi/AsadiC0GM21}.
Designing more efficient algorithms for solving constraints of this form would lead to better algorithms both for the safety problem studied in this work and for template-based affine inductive invariant synthesis in programs.

\begin{example}
	For completeness, we provide the constraints generated in Step~2 for \cref{ex:first_simple_case} with $N_I = 1$ for readability, i.e.\ our running example \cref{fig:running} with $\distribution_0 = \{A \mapsto \frac{1}{3}, B \mapsto \frac{1}{3}, C \mapsto \frac{1}{3}\}$ and $H = \{\distribution \mid \distribution(C) \geq \frac{1}{4}\}$, in \cref{fig:constraints_for_example}.
\end{example}

\begin{figure}[t]
	\begin{align*}
			\Phi_{\text{init}} & : \tpl{c_0} + \tpl{c_1} \cdot \tfrac{1}{3} + \tpl{c_2} \cdot \tfrac{1}{3} + \tpl{c_3} \cdot \tfrac{1}{3} \geq 0 \\
			\Phi_{\textrm{safe}} & : (\tpl{c_0} + \tpl{c_1} \cdot A + \tpl{c_2} \cdot B + \tpl{c_3} \cdot C \geq 0) \Longrightarrow C \geq \tfrac{1}{4} \\
			\Phi_{\textrm{inductive}} & : \begin{gathered}
					(\tpl{c_0} + \tpl{c_1} \cdot A + \tpl{c_2} \cdot B + \tpl{c_3} \cdot C \geq 0) \Longrightarrow \\ \tpl{c_0} + \tpl{c_1} \cdot (A \cdot \tpl{p_{A, a_1}} + \tfrac{1}{2} C) + \tpl{c_2} \cdot A \cdot \tpl{p_{A, a_2}} + \tpl{c_3} \cdot (B + \tfrac{1}{2} C) \geq 0
				\end{gathered} \\
			\Phi_{\textrm{strat}} & : \tpl{p_{A, a_1}} \geq 0 \quad \tpl{p_{A, a_2}} \geq 0 \quad \tpl{p_{A, a_1}} + \tpl{p_{A, a_2}} = 1
	\end{align*}
	\caption{
		List of constraints generated in Step~2 for \cref{ex:first_simple_case} with $N_I = 1$.
		The uppercase letters correspond to variables indicating the distribution in these states, i.e.\ $A$ refers to $\distribution(A)$.
		These also are the universally quantified variables, which will be handled by the quantifier elimination in Step~3.
		The template variables are written in grey.
		For readability, we omit the constraints required for state distributions $\mu \in \Distributions(\States)$, i.e.\ $A \geq 0$ etc.
		The actual query sent to the solver in Step~4 after quantifier elimination comprises 27 constraints with 21 variables.
	} \label{fig:constraints_for_example}
\end{figure}

To conclude this section, we emphasize that our algorithm \emph{simultaneously} synthesizes both the invariant and the witnessing strategy, which is the key component to achieve relative completeness.

\subsection{Synthesis of Affine Invariants and General Strategies}\label{sec:algotwo}
We now present our second algorithm, which additionally synthesizes \emph{distribution strategies} (of a particular shape) together with an affine inductive distributional invariant.
We refer to it as \algogen{}.
The second algorithm proceeds in the analogous four steps as the first algorithm, \algomem{}.
Hence, in the interest of space, we only discuss the differences compared to \algomem{}. 

\paragraph*{Step~1: Setting up Templates.}
The algorithm sets up templates for $\strategy$ and $I$.
The template for $I$ is defined analogously as in \cref{sec:algoone}.
However, as we now want to search for a strategy $\strategy$ that need not be memoryless but instead may depend on the current distribution, we need to consider a more general template.
In particular, the template for the probability $p_{\state_i,\action_j}$ of taking an action $\action_j$ in state $\state_i$ is no longer a constant value.
Instead, $p_{\state_i,\action_j}(\vec{x})$ is a function of the probability distribution $\vec{x}$ of the current state of the MDP, and we define its template to be a quotient of two affine expressions for each $s_i \in \States$ and $a_j \in \Actions(\state_i)$:
\begin{equation*}
	p_{\state_i, \action_j}(\vec{x}) \equiv \frac{\mathrm{num}(\state_i, \action_j)(\vec{x})}{\mathrm{den}(\state_i)(\vec{x})} \equiv \frac{r^{i,j}_0 + r^{i,j}_1 \cdot x_1 + \dots + r^{i,j}_n \cdot x_n}{s^{i}_0 + s^{i}_1 \cdot x_1 + \dots + s^{i}_n \cdot x_n}.
\end{equation*}
(In \cref{sec:extensions}, we discuss how to extend our approach to polynomial expressions for numerator and denominator, i.e.\ rational functions.)
Note that the coefficients in the numerator depend both on the state $\state_i$ and the action $\action_j$, whereas the coefficients in the denominator depend only on the state $\state_i$.
This is because we only use the affine expression in the denominator as a normalization factor to ensure that $p_{\state_i,\action_i}$ indeed defines a probability.

\paragraph*{Step~2: Constraint Collection.}
As before, the algorithm now collects the constraints over symbolic template variables which encode that $\strategy$ is a strategy, that $I$ is an inductive distributional invariant, and that $I$ is contained in $H$.
The constraints $\Phi_{\textrm{initial}}$, $\Phi_{\textrm{inductive}}$, and $\Phi_{\textrm{safe}}$ are defined analogously as in \cref{sec:algoone}, with the necessary adaptation to $\mathrm{step}(\vec{x})$.
For the strategy constraint $\Phi_{\textrm{strat}}$ we now need to take additional care to ensure that each quotient template defined above does not induce division by $0$ and that these values indeed correspond to a distribution over the available actions.
We ensure this by the following constraint:
\begin{equation*}
		\Phi_{\textrm{strat}} \equiv \forall \vec{x} \in \Reals^n.\ I(\vec{x}) \Longrightarrow {\bigwedge}_{i=1}^n \left(
		\begin{aligned}
			& {\bigwedge}_{\action_j \in \Actions(\state_i)} \mathrm{num}(\state_i, \action_j)(\vec{x}) \geq 0 \land {} \\
			& \mathrm{den}(\state_i)(\vec{x}) \geq 1 \land {} \\
			& {\sum}_{\action_j \in \Actions(\state_i)} \mathrm{num}(\state_i, \action_j)(\vec{x}) = \mathrm{den}(\state_i)(\vec{x}).
		\end{aligned} \right).
\end{equation*}
The first two constraints ensure that all quantities are positive and we never divide by $0$.
The third means that the numerators sum up to the denominator.
Together, this ensures the desired result, i.e.\ $p_{\state_i, \circ}(\vec{x}) \in \Distributions(\Actions(\state_i))$ whenever $\vec{x} \in \Distributions(\States)$.
Note that the $\geq 1$ constraint for the denominator can be replaced by an arbitrary constant $> 0$, since we can always rescale all involved coefficients.

\paragraph*{Step~3: Quantifier Elimination.}
The constraints $\Phi_{\textrm{strat}}$, $\Phi_{\textrm{initial}}$, and $\Phi_{\textrm{safe}}$ can be handled analogously to \cref{sec:algoone}.
In particular, by applying Farkas' lemma these can be 
translated into an equisatisfiable purely existentially quantified system of polynomial constraints, and our algorithm applies this translation.

However, the constraint $\Phi_{\textrm{inductive}}$ now involves quotients of affine expressions:
Upon splitting the conjunction on the right-hand side of the implication in $\Phi_{\textrm{inductive}}$ into a conjunction of implications, the inequalities on the right-hand side of these implications contain templates for strategy probabilities $p_{\state_i,\action_j}(\vec{x})$.
The algorithm removes the quotients by multiplying both sides of the inequality by denominators of each quotient.
(Recall that each denominator is positive by the constraint $\Phi_{\textrm{strat}}$.)
This results in the multiplication of symbolic affine expressions, hence $\Phi_{\textrm{inductive}}$ becomes a conjunction of implications of the form
\begin{equation*}
	\forall\vec{x}\in\mathbb{R}^n.\, (\textrm{affexp}_1(\vec{x}) \geq 0) \land \dots \land (\textrm{affexp}_N(\vec{x}) \geq 0) \Longrightarrow (\textrm{polyexp}(\vec{x}) \geq 0).
\end{equation*}
Here, each $\textrm{affexp}_i(\vec{x})$ is an affine expression over $\vec{x}$, but $\textrm{polyexp}(\vec{x})$ is now a polynomial expression over $\vec{x}$.
Hence we cannot apply a Farkas' lemma-style result to remove universal quantifiers.


Instead, we motivate our translation by recalling Handelman's theorem~\cite{handelman1988representing}, which characterizes \emph{strictly} positive polynomials over a set of affine inequalities.
It will allow us to soundly translate $\Phi_{\textrm{inductive}}$ into an existentially quantified system of constraints over the symbolic template variables, as well as fresh auxiliary variables that are introduced by the translation.

\begin{theorem}[\cite{handelman1988representing}] \label{thm:handelman}
	Let $\mathcal{X} = \{x_1,\dots,x_n\}$ be a finite set of real-valued variables, and consider the following system of $N\in\mathbb{N}$ non-strict affine inequalities over $\mathcal{X}$:
	\begin{equation*}\Phi:\, \begin{cases}
		c^1_0 + c^1_1\cdot x_1 + \dots + c^1_n \cdot x_n \geq 0 \\
		\qquad \qquad \qquad \vdots \\
		c^N_0 + c^N_1\cdot x_1 + \dots + c^N_n \cdot x_n \geq 0
	\end{cases}.
	\end{equation*}
	Let $\textrm{Prod}(\Phi) = \{ \prod_{i=1}^t \phi_i  \mid t \in \mathbb{N}_0, \phi_i \in \Phi \}$ be the set of all products of finitely many affine expressions in $\Phi$, where the product of $0$ affine expressions is a constant expression $1$.
	Suppose that $\Phi$ is satisfiable and that $\{\vec{y} \mid \vec{y} \models \Phi\}$, the set of values satisfying $\Phi$, is topologically compact, i.e.\ closed and bounded.
	Then $\Phi$ entails a polynomial inequality $\phi(\vec{x}) > 0$ if and only if $\phi$ can be written as a non-negative linear combination of finitely many products in $\textrm{Prod}(\Phi)$, i.e.\ if and only if there exist $y_1,\dots,y_n\geq 0$ and $\phi_1,\dots,\phi_n\in\textrm{Prod}(\Phi)$ such that $\phi = y_1 \cdot \phi_1 + \dots + y_n \cdot \phi_n$.
\end{theorem}

Notice that we cannot directly apply Handelman's theorem to a constraint
\begin{equation*}
	\forall \vec{x} \in \mathbb{R}^n.\ (\textrm{affexp}_1(\vec{x}) \geq 0) \land \dots \land (\textrm{affexp}_N(\vec{x}) \geq 0) \Longrightarrow (\textrm{polyexp}(\vec{x}) \geq 0),
\end{equation*}
since the polynomial inequality on the right-hand-side of the implication is non-strict whereas the polynomial inequality in Handelman's theorem is strict.
However, the direction needed for the soundness of translation holds even with the non-strict polynomial inequality on the right-hand side.
In particular, it clearly holds that if $\textrm{polyexp}$ can be written as a non-negative linear combination of finitely many products of affine inequalities, then $\textrm{polyexp}$ is non-negative whenever all affine inequalities are non-negative.
Hence, we may use the translation in Handelman's theorem to translate each implication in $\Phi_{\textrm{inductive}}$ into a system of purely existentially quantified constraints.

As Handelman's theorem does not impose a bound on the number of products of affine expressions that might appear in the translation, 
we \emph{parametrize} the algorithm with an upper bound $K$ on the maximal number of affine inequalities appearing in each product.
To that end, we define $\textrm{Prod}_K(\Phi) =  \{\prod_{i=1}^t \phi_i \, \mid \, 0\leq t \leq K,\, \phi_i\in\Phi\}$.
Let $M_K = |\textrm{Prod}_K(\Phi)|$ be the total number of such products and $\textrm{Prod}_K(\Phi) = \{\phi_1,\dots,\phi_{M_K}\}$.
%
%
Then, for any constraint of the form
\begin{equation*}
	\forall\vec{x}\in\mathbb{R}^n.\, (\textrm{affexp}_1(\vec{x}) \geq 0) \land \dots \land (\textrm{affexp}_N(\vec{x}) \geq 0) \Longrightarrow (\textrm{polyexp}(\vec{x}) \geq 0),
\end{equation*}
we introduce fresh template variables $y_1,\dots,y_{M_K}$ and translate it into the system of purely existentially quantified constraints
\begin{equation*}
	(y_1 \geq 0) \land \dots \land (y_N \geq 0) \land (\textrm{polyexp}(\vec{x}) \equiv_H y_1 \cdot \phi_1(\vec{x}) + \dots + y_{M_K} \cdot \phi_{M_K}(\vec{x})).
\end{equation*}
Here, $\textrm{polyexp}(\vec{x}) \equiv_H y_1 \cdot \phi_1(\vec{x}) + \dots + y_{M_K} \cdot \phi_{M_K}(\vec{x})$ denotes the set of equalities over template variables and $y_1,\dots,y_{M_K}$ which equate the constant coefficients as well as the coefficients of each monomial over $\{x_1,\dots,x_k\}$ of degree at most $K$ on two sides of the equivalence, as specified by Handelman's theorem.

While our translation into a purely existentially quantified constraints is not complete due to the non-strict polynomial inequality and due to the parametrization by $K$, Handelman's theorem justifies the translation as it indicates that the translation is \enquote{close to complete} for sufficiently large values of $K$.

\paragraph*{Step~4: Constraint Solving.}
This step is analogous to \cref{sec:algoone} and we use an off-the-shelf polynomial constraint solver to handle the resulting system of purely existentially quantified polynomial constraints.
If the solver outputs a solution, we conclude that the computed $I$ is an inductive distributional invariant for the computed strategy $\strategy$ and initial distribution $\distribution_0$, and that $I$ is contained in $H$.
Therefore, by \cref{thm:inv}, we conclude that $\distribution_0$ is $H$-safe under $\strategy$.

\begin{theorem} \label{thm:secondalgo}
	\emph{Soundness}: Suppose \algogen{} returns a strategy $\strategy$ and an affine inductive distributional invariant $I$.
	Then, $\strategy$ is $H$-safe for $\distribution_0$.

	\emph{Complexity}: For any fixed parameter $K \in \Naturals$, the runtime of \algogen{} is in PSPACE in the size of the MDP and the template size parameter $N_I \in \Naturals$.
\end{theorem}
\iftoggle{arxiv}{

\begin{proof}
	Soundness follows from the fact that Step~2 encodes all defining constraints of strategies, affine inductive distributional invariants, and initial probability distributions as constraints, that Step~3 soundly converts the constraints into purely existentially quantified system of constraints, and that $I \subseteq H$ if and only if $\distribution_0$ is $H$-safe under $\strategy$ by \cref{thm:inv}.
	
	The runtime complexity claim follows from the fact that the first three steps of the algorithm all have polynomial runtime and thus yield a system of constraints which is polynomial in the size of the MDP and the template size parameter $N_I\in \mathbb{N}$ (note that the value of the parameter $K$ is assumed to be fixed).
	Thus, as the existential first-order theory of the reals is in PSPACE, it follows that solving the resulting system of polynomial constraints can be done in PSPACE. \qed
\end{proof}
}{The proof can be found in \cite[Sec.\ 5.2]{arxiv}.}

\section{Discussion, Extensions, and Variants} \label{sec:extensions}
With our two algorithms in place, we remark on several interesting details and possibilities for extensions.

\paragraph*{Polynomial Expressions.}
Our second algorithm can also be extended to synthesizing \emph{polynomial} inductive distributional invariants, i.e.\ instead of defining the invariant $I$ through a conjunction of affine linear expressions we could synthesize polynomial expressions such as $x_1^2 + x_2 \cdot x_3 \leq 0.5$.
This can be achieved by using Putinar's Positivstellensatz~\cite{putinar1993positive} instead of Handelman's theorem in Step~3.
This technique has recently been used for generating polynomial inductive invariants in programs in~\cite{DBLP:conf/pldi/Chatterjee0GG20}, and our translation in Step~3 can be analogously adapted to synthesize polynomial inductive distributional invariants up to a specified degree.
In the same way, instead of requiring that $H$ is given as a conjunction of affine linear constraints, we can also handle the case of polynomial constraints.
The same holds true for the probabilities of choosing certain actions $p_{\state_i, \action_j}(\vec{x})$.
While we have defined these as fractions of affine linear expressions, we could replace them with rational functions.

We chose to exclude treatment of this case for the sake of readability.

\paragraph*{Uninitialized and Restricted Initial Case.}
We remark that we can directly incorporate the uninitialized case in our algorithm.
In particular, instead of requiring that $I(\distribution_0)$ holds for the concretely given initial values, we can instead existentially quantify over the values of $\distribution_0(\state_i)$ and add the constraint that $\distribution_0$ is a distribution, i.e.\ $\distribution_0(\state_i) \in \Distributions(\States)$.
This does not add universal quantification, thus we do not need to apply any quantifier elimination for these variables.
This also subsumes and generalizes the ideas of \cite{DBLP:conf/lics/AkshayGV18}, which observes that checking whether a fixpoint of the transition dynamics lies within $H$ is sufficient.
Choosing $I = \{\distribution^*\}$ where $\distribution^*$ is such a fixpoint satisfies all of our constraints.
\iftoggle{arxiv}{
We can also adapt our constraints to only consider these fixpoints, as follows.
First, \cite[Lemma~3.4]{DBLP:conf/lics/AkshayGV18} shows that considering distributions which are fixpoints under one-step strategies is sufficient.
Note that the proof for this lemma does not rely on $H$ being a polytope, but only on it being closed and convex.
Thus, we can consider the following constraints:
\begin{gather*}
	\Phi_{\mathrm{safe}}: \distribution \in \Distributions(\States) \land \distribution \in H \\
	\Phi_{\mathrm{strat}}: {\bigwedge}_{i=1}^n \left( p_{\state_i, \circ} \in \Distributions(\Actions(\state_i)) \right) \\
	\Phi_{\mathrm{fix}}: \distribution = \mathrm{step}(\distribution).
\end{gather*}
Observe that all occurring variables are existentially quantified, thus we do not need to perform quantifier elimination.
We mention several consequences:
First, as long as we can write $\distribution \in H$ in the existential theory of the reals and $H$ is closed and convex, our approach is applicable.
Second, we also get out a witness strategy, however we also pay a price in terms of complexity, since this algorithm lies in $\exists \Reals$, compared to the PTIME approach of \cite{DBLP:conf/lics/AkshayGV18}.

However, observe that the constraints of $\Phi_{\mathrm{fix}}$ can be written as follows.
For a fixed transition $(s_i, a_j, s')$, we get that the probability mass moving through this transition is given by $\distribution(s_i) \cdot \transitions(s_i, a_j, s') \cdot p_{s_i, a_j}$.
By assigning this value to an intermediate variables $y_{s_i, a_j, s'}$ and equating $\distribution(s') = \sum_{(s_i, a_j)} y_{s_i, a_j, s'}$, $\Phi_{\mathrm{fix}}$ can be written as a quadratic constraint $\hat{\distribution}^T \cdot Q \cdot \hat{p} = \vec{y}$.
The matrix $Q$ is a diagonal matrix with only positive entries on the diagonal, one for each transition.
The row vector $\hat{\distribution}^T$ comprises $\distribution(s_i)$, with an entry at position $k$ corresponding to the source state of the $k$-th state-action pair in $Q$.
(We can ensure that the duplicated entries all equal $\distribution(s_i)$ by equality constraints.)
The column vector $\hat{p}$ contains the probability of playing the action corresponding to the $k$-th transition, again duplicated and equated where required.
Observe that $Q$ is positive definite and the overall constraints are of polynomial size (at most three equation per transition).
If $H$ is again given through linear inequalities, we can encode all our constraints as quadratic program.
Since $Q$ is positive definite, we can determine satisfiability in polynomial time \cite{DBLP:journals/mp/YeT89}, recovering the complexity result of \cite{DBLP:conf/lics/AkshayGV18}.

%
%
%
%
}{
See \cite[Sec.\ 6]{arxiv} for details.
}

Our algorithm is also able to handle the \enquote{intermediate} case, as follows.
The uninitialized case leaves absolute freedom in the choice of initial distribution, while the initialized case concretely specifies one initial distribution.
Here, we could as well impose \emph{some} constraints on the initial distribution without fixing it completely, i.e.\ ask whether there exists an $H$-safe initial distribution $\distribution_0$ which satisfies a predicate $\Phi_{\text{init}}$.
If $\Phi_{\text{init}}$ is a conjunction of affine linear constraints, we can directly handle this query, too.
Note that both initialized and uninitialized are special cases thereof.

\paragraph*{Non-Inductive Initial Steps.}
Instead of requiring to synthesize an invariant which contains the initial distribution, we can explicitly write down the first $k$ distributions and only then require an invariant and strategy to be found.
More concretely, the set of distributions that can be achieved in a given step $k$ while remaining in $H$ can be explicitly computed, denote this set as $\Delta^k$.
For a different perspective, this describes the set of states reachable in $\TS<\MDP>$ within $k$ steps and corresponds to \enquote{unrolling} the MDP for a fixed number of steps.
This then goes hand in hand with the above \enquote{restricted initial case}, where we ask whether there exists an $H$-safe distribution in $\Delta^k$.
We conjecture that this could simplify the search for distributional invariants for systems which have a lot of \enquote{transient} behaviour, as observed in searching for invariants for state reachability~\cite{BatzCKKMS20}.


\section{Implementation and Evaluation}\label{sec:experiments}
While the main focus of our contribution lies on the theory, we validate the applicability through an unoptimized prototype implementation.
We implemented our approach in Python 3.10, using \texttt{SymPy} 1.11 \cite{DBLP:journals/peerj-cs/MeurerSPCKRKIMS17} to handle and simplify symbolic expressions, and \texttt{PySMT} 0.9 \cite{gario2015pysmt} to abstract communication with constraint solvers.
We use \texttt{z3} 4.8 \cite{DBLP:conf/tacas/MouraB08} and \texttt{mathsat} 5.6 \cite{mathsat5} as back-ends.
Our experiments were executed on consumer hardware (AMD Ryzen 3600 CPU with 16 GB RAM).

\begin{table}[t]
	\caption{
		Overview of our results for the five considered models.
		From left to right, we list the name of the model, the runtime, and size of the invariant, followed by the number of variables, constraints, and total size of the query passed to the constraint solvers.
		For $\mathsf{Running}$, we provided additional hints to the solver to achieve a more consistent runtime, indicated by the dagger symbol.
	} \label{tbl:results} %
	\centering %
	\setlength{\tabcolsep}{4pt} %
	\begin{tabular}{cccccc}
		Model                   &   Runtime    & $N_I$ & \#Var. & \#Constr. & Size. \\
		\midrule
		$\mathsf{Running}$            & 3s$^\dagger$ &   3   &   92   &    123    &  849  \\
		$\mathsf{Chain}$             &     10s      &   2   &   69   &    82     &  666  \\
		$\mathsf{Split}$             &      3s      &   3   &   60   &    69     &  571  \\
		$\mathsf{PageRank}$            &      3s      &   2   &   44   &    52     &  536  \\
		$\mathsf{Insulin}\text{-}^{131}\text{I}$ &      2s      &   2   &   44   &    52     &  476
	\end{tabular}
\end{table}

\paragraph*{Caveats.}
While the existential (non-linear) theory of the reals is known to be decidable, practical algorithms are less explored than, for example, SAT solving.
In particular, runtimes are quite sensitive to minor changes in the input structure and initial randomization (many solvers apply randomized algorithms).
We observed differences of several orders of magnitude (going from seconds to hours) simply due to restarting the computation (leading to different initial seeds).
Similarly, by strengthening the antecedents of implications by known facts, we also observed significant improvements.
Concretely, given that we have constraints of the form $I(\vec{x}) \Longrightarrow H(\vec{x})$ and $I(\vec{x}) \Longrightarrow \Phi(\vec{x})$, we observed that changing the second constraint to $I(\vec{x}) \land H(\vec{x}) \Longrightarrow \Phi(\vec{x})$ would drastically improve the runtime even though the two are semantically equivalent.

This suggests that both improvements of our implementation as well as further work on constraint solvers are likely to have a significant impact on the runtime.

\paragraph*{Models.}
Aside from our running example of \cref{fig:running}, which we refer to as $\mathsf{Running}$ here, we consider two further toy examples.

The first model, called $\mathsf{Chain}$, is a Markov chain defined as follows:
We consider the states $S = \{s_1, \dots, s_{10}\}$ and set $\transitions(s_i) = \{s_{i+1} \mapsto 1\}$ for all $i < 10$ and $\transitions(s_{10}) = \{s_9 \mapsto \frac{1}{2}, s_{10} \mapsto \frac{1}{2}\}$.
The initial distribution is given as $\distribution_0(s_i) = \frac{1}{10}$ for all $s_i \in \States$ and the safe set by $H = \{\distribution(s_{10}) \geq \frac{1}{10}\}$.
We are mainly interested in this model to investigate demonstrate applicability to \enquote{larger} systems.

The second model, called $\mathsf{Split}$, is an MDP which actually comprises two independent subsystems.
We depict the model in \cref{fig:model_split}.
The initial distribution is $\distribution_0 = \{A \mapsto \frac{1}{2}, C \mapsto \frac{1}{2}\}$ and the safe set $H = \{\distribution(A) + \distribution(D) \geq \frac{1}{2}\}$.
This aims to explore both disconnected models as well as a safe set which imposes a constraint on multiple states at once.
In particular, observe that initially $\distribution_0(D) = 0$ but $\distribution_i(D)$ converges to $1$ while $\distribution_i(A)$ converges to $0$, even if choosing action $a_1$.
Thus, the invariant needs to identify the simultaneous flow from $A$ to $B$ and $C$ to $D$.

We additionally consider two examples from the literature, namely the $\mathsf{PageRank}$ example from \cite[Fig.~3]{DBLP:journals/jacm/AgrawalAGT15}, based on \cite{DBLP:books/daglib/0018460}, and $\mathsf{Insulin}\text{-}^{131}\text{I}$, a pharmacokinetics system \cite[Example~2]{DBLP:journals/jacm/AgrawalAGT15}, based on \cite{DBLP:conf/qest/ChadhaKVAK11}.
Both are Markov chains.

\begin{figure}[t]
	\centering
	\begin{tikzpicture}[auto]
		\node[state] at (0, 0) (A) {A};
		\node[state] at (2, 0) (B) {B};
		\node[state] at (4, 0) (C) {C};
		\node[state] at (6, 0) (D) {D};
		\node[actionnode] at (0.5,0.5) (a1) {};

		\draw[directedge]
			(A) edge[swap,bend right=10] node[action] {$a_2$} (B)
			(B) edge[loop right] (B)
			(D) edge[loop right] (D)
		;
		\draw[actionedge]
			(A) edge node[action,pos=0.9,swap] {$a_1$} (a1)
		;
		\draw[probedge]
			(a1) edge[out=45,in=100,looseness=1.5,swap] node[prob] {$0.9$} (A)
			(a1) edge[out=45,in=120,looseness=0.8] node[prob] {$0.1$} (B)

			(C) edge[out=43,in=90,looseness=3,swap,pos=0.55] node[prob] {$0.5$} (C)
			(C) edge[out=40,pos=0.50,in=140] node[prob] {$0.5$} (D)
		;
	\end{tikzpicture}
	\caption{
		Our $\mathsf{Split}$ toy example.
		The MDP comprises two disconnected parts.
		Probability mass flows from $A$ to $B$ and from $C$ to $D$ under all strategies.
	} \label{fig:model_split}
\end{figure}
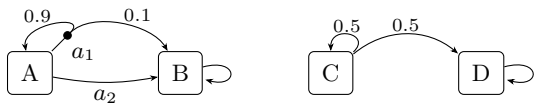

\paragraph*{Results.}
We summarize our findings briefly in \cref{tbl:results}.
We again underline that not too much attention should be put on runtimes, since they are very sensitive to minimal changes in the model.
The evaluation is mainly intended to demonstrate that our methods are actually able to provide results.
For completeness, we report the size of the invariant $N_I$ and the size of the constraint problem in terms of number of variables, constraints, and operations inside these constraints.
We also provide the invariants and strategy identified by our method in \cref{tbl:invariants_strategies}.
Note that for $\mathsf{Running}$ we used \algogen, while the other two examples are handled by \algomem. For $\mathsf{Running}$, we observed a significant dependence on the initialization of the solvers.
Thus we added several \enquote{hints}, i.e.\ known correct values for some variables.
(To be precise, we set the value for eight of the 92 variables.)

\begin{table}[t]
	\caption{
		The invariants and strategies computed for our models.
		We omit the invariants for the two real-world scenarios since they are too large to fit.
	} \label{tbl:invariants_strategies} %
	\centering %
	\renewcommand{\arraystretch}{1.5} %
	\begin{tabular}{cc}
		Model & Computed Invariant and Strategy \\
		\midrule
		$\mathsf{Running}$ & $\{A \geq \frac{1}{4}, B = \frac{1}{4}\}$ \qquad $\strategy(\distribution) = \{a_1 \mapsto \frac{1}{4 \cdot \distribution(A)}, a_2 \mapsto \frac{4 \cdot \distribution(A) - 1}{4 \cdot \distribution(A)}\}$ \\
		$\mathsf{Chain}$ & $\{s_9 + s_{10} \geq \frac{1}{5}, s_{10} \geq \frac{1}{10}\}$ \qquad $\strategy = \emptyset$ (Markov chain) \\
		$\mathsf{Split}$ & $\{B \leq D, A + B \geq C + D, 3 \cdot (C + D) - (A + B) \geq 1\}$ \qquad $\strategy = \{a \mapsto 1\}$
	\end{tabular}
\end{table}

\paragraph*{Discussion.}
We remark two related points:
Firstly, we observe that very often most of the involved auxiliary variables introduced by the quantifier elimination have a value of zero. 
Thus, a potential optimization is to explicitly set most such variables to zero, check whether the formula is satisfiable, and, if not, gradually remove these constraints either at random or guided by unsat-cores if available (i.e.\ clauses which are the \enquote{reason} for unsatisfiability).
Moreover, we observed significant differences between the solvers:
While \texttt{z3} seems to be much quicker to identify unsatisfiability, \texttt{mathsat} usually is better at finding satisfying assignments.
Hence, using both solvers in tandem seems to be very beneficial.

\section{Conclusion}\label{sec:conclusion}
We developed a framework for defining certificates for safety objectives in MDPs as distributional inductive invariants. Using this, we came up with two algorithms that synthesize linear/affine invariants and corresponding memoryless or general strategies for safety in MDPs. To the best of our knowledge this is the first time the template-based invariant approach, already known to be successful for programs, has been applied to synthesis strategies in MDPs for distributional safety properties. Further, our experimental results show that our affine invariants are sufficient for many interesting examples. However, the second approach can in fact be lifted to synthesize polynomial invariants, and hence potentially, a large set of MDPs. Exploring this could be a future line of work. Yet another avenue would be to lift this work to more complex objectives. It would also be interesting to explore how one can automate distributional invariant synthesis if the safe set $H$ is specified in terms of both strict and non-strict inequalities, while preserving completeness guarantees. Finally, in terms of applicability, we would like to apply this approach to solve more benchmarks and problems, e.g., to synthesize risk-aware strategies for MDPs~\cite{DBLP:conf/aaai/Meggendorfer22,DBLP:conf/lics/KretinskyM18}.


\clearpage
\bibliographystyle{splncs04}
\bibliography{main}

\end{document}